\documentclass[reqno]{amsart}
\usepackage{amsthm}
\usepackage{amstext}
\usepackage{bbm}
\usepackage{enumitem}
\usepackage{amssymb}
\usepackage{quiver, adjustbox}
\usepackage{amsmath,calligra,mathrsfs}
\usepackage[all,cmtip]{xy}
\usepackage{dsfont}
\usepackage{hyperref}
\usepackage{graphicx}

\usepackage{colonequals}

\newcommand\bb{\mathbb}

\makeatletter
\renewcommand*\env@matrix[1][*\c@MaxMatrixCols c]{%
  \hskip -\arraycolsep
  \let\@ifnextchar\new@ifnextchar
  \array{#1}}
\makeatother

\theoremstyle{plain}
\newtheorem{thm}{Theorem}[section]
\newtheorem{lem}[thm]{Lemma}
\newtheorem{prop}[thm]{Proposition}

\theoremstyle{definition}
\newtheorem{defn}[thm]{Definition}
\newtheorem{exmp}[thm]{Example}

\newtheorem{rem}[thm]{Remark}
\theoremstyle{plain}

\DeclareMathOperator{\sheafhom}{\mathscr{H}\text{\kern -3pt {\calligra\large om}}\,}

\DeclareMathOperator{\Rbb}{\mathbb{R}}
\DeclareMathOperator{\Zbb}{\mathbb{Z}}
\DeclareMathOperator{\ECT}{ECT}

\DeclareMathOperator{\QECT}{QECT}

\DeclareMathOperator{\CF}{CF}

\usepackage{lipsum}
\usepackage{fullpage}
\usepackage{amsthm, thmtools}
\usepackage{parskip}

\usepackage{enumitem}

\renewenvironment{proof}{%
    \vspace{-\parskip}\begin{oldproof}%
    }{%
    \end{oldproof}\vspace{-\parskip}%
}

\let\phi\varphi
\usepackage{comment}
\usepackage{xcolor,mathtools}
\usepackage[authoryear, square, numbers]{natbib}
\setcitestyle{round}

\title{On the Injectivity of Euler Integral Transforms with Hyperplanes and Quadric Hypersurfaces}
    \author[1]{Mattie Ji}
    \address{Brown University, Department of Mathematics, Box 1917, 151 Thayer Street, Providence, RI 02912, USA}
    \email{mattie\_ji@brown.edu}

\begin{document}
\begin{abstract}
    The Euler characteristic transform (ECT) is an integral transform used widely in topological data analysis. Previous efforts by Curry et al. and Ghrist et al. have independently shown that the ECT is injective on all compact definable sets. In this work, we first study the injectivity of the ECT on definable sets that are not necessarily compact and prove a complete classification of constructible functions that the Euler characteristic transform is not injective on. We then introduce the quadric Euler characteristic transform (QECT) as a natural generalization of the ECT by detecting definable shapes with quadric hypersurfaces rather than hyperplanes. We also discuss some criteria for the injectivity of $\QECT$.
\end{abstract}


\maketitle

\section{Introduction}\label{sec:Introduction}

The Euler characteristic transform (ECT) is an integral transform in topological data analysis (TDA) introduced in \cite{turner2014frechet}. Since then, the ECT itself and its variants have been widely used in the applied science \citep{phenotype, crawford2020predicting, wang2021statistical, 10208939, hacquard2023euler, roell2023differentiable}. On a high level, the ECT takes in a shape $S$ in $\bb R^n$, ``scans" $S$ through each direction $v \in \bb S^{n-1}$, and keeps track of the Euler characteristics of the sublevel sets of $S$. Concretely, the Euler characteristic transform of $S$ may be formalized into a function $\ECT(S): \bb S^{n-1} \times \bb R \to \bb Z$ defined by
\[(\nu, t) \mapsto \ECT(S)(\nu, t) = \chi(\{x \in S\ |\ x \cdot \nu \leq t\}),\]
where $\chi(\bullet)$ denotes the combinatorial Euler characteristic (see Definition~\ref{def: comb_euler_chars}). Note that $x \cdot \nu = t$ defines the equation of a hyperplane in $\bb R^n$.

In the conclusion of \cite{curry2022many}, the authors posed the question of how the Euler characteristic transform (ECT) would behave on shapes cut out by quadratic equations rather than linear equations. Inspired by this question, we consider a ``converse" of this question in this work - \textit{what if we replace hyperplanes in the $\ECT$ with quadric hypersurfaces?}

The equation of a quadric hypersurface may be written as $x^T A x + \nu \cdot x = t$, where $A$ is a symmetric $n \times n$ real matrix and $\nu$ is a vector in $\bb R^n$. Based on this notion, we can define the quadric Euler characteristic transform (QECT) of the shape $S$ as a function given by $(A, \nu, t) \mapsto \QECT(S)(A, \nu, t) \coloneqq \chi(\{x \in S\ |\ x^T A x + x \cdot \nu \leq t\}$. This definition will be made more precise in Section~\ref{sec:Quadric Hypersurfaces}. By extending the class of hyperplanes to quadric hypersurfaces, the hope is that the $\QECT$ would add an extra variable that takes into account of curvatures.

In this work, a central question we are interested in is the injectivity of the $\ECT$ and the $\QECT$. Previous efforts by \cite{curry2022many} and \cite{ghrist2018persistent} have independently shown that the $\ECT$ is injective on a ``reasonable" class of compact shapes. Furthermore, the work in \cite{ghrist2018persistent} showed that this injectivity result extends to finite sums of indicator functions on a collection of shapes (known as \textit{constructible functions}) that are compactly supported.

We will first investigate the injectivity of the ECT on constructible functions that are not compactly supported. As we will see in Example~\ref{ex: non-injective definable sets} and Theorem~\ref{thm: hyperplane injective} in Section~\ref{sec:Hyperplanes}, there are many pairs of constructible functions that the $\ECT$ is not injective on. We then extend the ECT to the QECT and discuss its injectivity in Theorem~\ref{thm: quadric zero vector} and Theorem~\ref{thm: quadric interpolation}. Specifically, we will prove the following main results in this paper.
\begin{enumerate}
    \item We completely classify all the pairs of constructible functions that the Euler characteristic transform is not injective on in Theorem~\ref{thm:injectivity_hyperplane}.
    \item Suppose $v = 0$ is fixed, we show that the function $S \mapsto \{(A, t) \mapsto \QECT(S)(A, 0, t)\}$ is injective up to sign in Theorem~\ref{thm: quadric zero vector}.
    \item Suppose the classes of ``reasonable" shapes (see Definition~\ref{def: o-minimal}) we are considering are all contained in $B_R(0) \coloneqq \{x \in \bb R^n\ |\ |x| \leq R\}$ for some $R \geq 0$. For a fixed $A$ such that $||A||_{op} < \frac{1}{1 + 2 R^2}$, we show that the function $S \mapsto \{(v, t) \mapsto \QECT(S)(A, v, t)\}$ is injective in Theorem~\ref{thm: quadric interpolation}. In particular, this serves as an interpolation between the injectivity of the $\ECT$ and Theorem~\ref{thm: quadric zero vector} (see Remark~\ref{rem: interpolation}).
\end{enumerate}
These statements will be made more precise in their respective theorems.

\subsection{Outline} The paper is organized as follows. In Section~\ref{sec:Background}, we introduce the relevant backgrounds in o-minimal structures, Euler calculus, and the ECT. In Section~\ref{sec:Hyperplanes}, we discuss the injectivity of the ECT to all constructible functions, leading to a complete characterization of injectivity in Theorem~\ref{thm: hyperplane injective}. In Section~\ref{sec:Quadric Hypersurfaces}, we extend the ECT to the QECT by considering quadric surfaces rather than hyperplanes in the sublevel sets of the integral transform and discuss several results on the injectivity of the QECT, leading to Theorem~\ref{thm: quadric zero vector} and Theorem~\ref{thm: quadric interpolation} in the end.

\subsection*{Acknowledgements} M.J. would like to thank Professor Kun Meng and Professor Richard Schwartz for their helpful comments and discussions. M.J. would also like to thank Cheng Chen for helpful conversations on functional analysis. M.J. would also like to thank Nir Elber and Riley Guyett for proofreading the paper and providing feedback and suggestions.

\section{Background}\label{sec:Background}

In this section, we will cover the necessary backgrounds in o-minimal structures, Euler calculus, and the ECT. We refer the reader to \cite{van1998tame} for a comprehensive introduction to o-minimal structures, \cite{curry2012euler} and \cite{Gusein-Zade_2010} for more details in Euler calculus, and \cite{munch2023invitation} for a general review of the Euler characteristic transform.

\subsection{O-minimal structures}

O-minimal structures are widely used as the mathematical representation of a shape \citep{ghrist2014elementary, curry2022many, kirveslahti2023representing, meng2023Inference} in applied topology and topological data analysis. Often in integral geometry, we want to consider shapes that have some level of ``tameness" to avoid pathological examples, and o-minimal structures offer one way to capture the idea of ``tameness".

\begin{defn}\label{def: o-minimal}
    Let $\mathcal{O}_n$ be a collection of subsets of $\Rbb^n$ and $\mathcal{O} = \{\mathcal{O}_n\}_{n \geq 1}$, we say that $\mathcal{O}$ is an \textit{o-minimal structure} if it satisfies the following seven axioms.
\begin{enumerate}
        \item $\mathcal{O}_n$ is a Boolean algebra.
        \item If $A \in \mathcal{O}_n$, then $A \times \Rbb \in \mathcal{O}_{n+1}$ and $\Rbb \times A \in \mathcal{O}_{n+1}$.
        \item $\{(x_1, ..., x_n) \in \Rbb^n\ |\ x_i = x_j\} \in \mathcal{O}_n$ for $1 \leq i < j \leq n$.
        \item $\mathcal{O}$ is closed under axis-aligned projections.
        \item $\{r\} \in \mathcal{O}_1$ for all $r \in \Rbb$ and $\{(x, y) \in \Rbb^2\ |\ x < y\} \in \mathcal{O}_2$.
        \item $\mathcal{O}_1$ is exactly the finite unions of points and open intervals.
        \item $\mathcal{O}$ contains all real algebraic sets.
\end{enumerate}
An element of $\mathcal{O}_n$ is called a \textit{definable} set.
\end{defn}

In this paper, we will fix an arbitrary $o$-minimal structure $\mathcal{O}_n$. In particular, Definition~\ref{def: o-minimal} implies that any o-minimal structure has to contain all semialgebraic sets (see Remark 2.2 of \cite{curry2022many}). We also want a notion of ``definability" for functions between definable sets.

\begin{defn}\label{def: definable functions}
    Let $f: X \to Y$ be a function between definable spaces.
    \begin{enumerate}
        \item $f$ is called \textit{definable} if its graph is a definable set.
        \item If $f$ is continuous definable with continuous definable inverse, then $f$ is called a \textit{definable homeomorphism}, and $X$ and $Y$ are said to be \textit{definably homeomorphic}.
        \item If $f$ is an integer valued function, then $f: X \to \bb Z$ is called a \textit{constructible function}. Let $\operatorname{CF}(X)$ denote the space of constructible functions on $X$. Note that Definition~\ref{def: o-minimal}(4) implies that the image of $f$ is a discrete definable subset of $\mathcal{O}_1$ and is thus finite.
    \end{enumerate}
\end{defn}

When we define the quadric Euler characteristic transform later, we want to consider a suitable norm on the space of symmetric $n \times n$ matrices. There are many choices of norms for matrices that are popular in machine learning, such as the Schatten norm, cut norms, and $L_{p, q}$ norms, or the operator norm (see \cite{Fan_Li_Zhang_Zou_2020}), so we need to consider norms that are compatible with our $o$-minimal structure.

\begin{defn}\label{def: definable norm}
Let $(V, ||\bullet||)$ be a finite-dimensional normed real vector space.
\begin{enumerate}
    \item $||\bullet||$ is called a definable norm if the norm function $||\bullet||: V \to \bb R_{\geq 0}$ is definable.
    \item $S^V \coloneqq \{x \in V\ |\ ||x|| = 1\}$ is called the unit sphere with respect to $||\bullet||$. Note that $S^V$ is a compact definable set when $||\bullet||$ is a definable norm.
    \item In particular, we will use $|\bullet|$ to denote the usual $\ell_2$ norm on $\bb R^n$, and $\mathbb{S}^{n-1}$ to denote the usual unit sphere in $\bb R^n$ with respect to the $\ell_2$ norm.  
\end{enumerate}
\end{defn}

\begin{exmp}
Here are some examples of definable norms that will be relevant to our discussions in Section~\ref{sec:Quadric Hypersurfaces}.
    \begin{enumerate}
        \item Let $|\bullet|$ be the $\ell_2$ norm on $\bb R^n$, then the set
        \[\{(x_1, ..., x_n, y) \in \bb R^n \times \bb R\ |\ x_1^2 + ... + x_n^2 = y^2 \text{ and } y \geq 0\}\]
        is a semialgebraic set and is hence definable. This is the graph of $|\bullet|: \bb R^n \to \bb R_{\geq 0}$.
        \item Let $V$ be the vector space of $n \times n$ symmetric real matrices and $||\bullet||_{\mathrm{op}}$ be the operator norm on $V$. The graph of $||\bullet||_{op}: V \to \bb R_{\geq 0}$ may be realized as an axis-aligned projection of the following definable set
        \[\{(A, \lambda_1, ..., \lambda_n) \in V \times \bb R^n\ |\ \det(A - \lambda_i I) = 0, e_i(\lambda_1, ..., \lambda_n) = \frac{a_i(A)}{a_n(A)} \text{ for all $i = 1, ..., n$} \text{ and } \lambda_1^2 \leq ... \leq \lambda_n^2\},\]
        where $e_i$ denotes the $i$-th elementary symmetric polynomial, $a_i(A)$ denotes the $i$-th coefficient of the characteristic polynomial $\det(A - xI)$. Note that $a_i$ is a polynomial function on the components of the matrix $A$. Thus, the operator norm is definable.
    \end{enumerate}
\end{exmp}

We also state the following technical lemma on o-minimal structures that will be used later in the paper.
\begin{lem}[Rephrased from Proposition 2.10 of Chapter 4 of \cite{van1998tame}]\label{lemma:: finite euler char types}
Let $S \subseteq \mathbb{R}^{m+n}$ be a definable set. For any $a\in\mathbb{R}^m$, define $S_a :=\{x\in\mathbb{R}^n\vert\, (a,x)\in S\}$. Then $\chi(S_a)$ takes only finitely many values as $a$ runs through $\bb R^m$, and for each integer $e$ the set $\{a \in \bb R^m: \chi(S_a) = e\}$ is definable.
\end{lem}

\subsection{Euler Calculus and the Euler Characteristic Transform}

Let $S \subseteq \Rbb^n$ be a definable set, the cell decomposition theorem \citep[][Chapter 3, Theorem 2.11]{van1998tame} asserts that there is a disjoint partition of $S$ into open-cells $C_1, ..., C_N$ such that each $C_i$ is definably homeomorphic to $\Rbb^{a_i}$ for some $a_i$. 
\begin{defn}\label{def: comb_euler_chars}
        Choose $S$ as above, the \textit{Euler characteristic} of $S$ is $\chi(S) \coloneqq \sum_{i = 1}^N (-1)^{a_i}$. This quantity is independent of the cell partition and is preserved under definable homeomorphisms (see Chapter $4$ of \cite{van1998tame}).
\end{defn}

Euler calculus is an integral calculus based on the observation that the \textit{Euler characteristics} $\chi(\bullet)$ exhibits a finitely additive property similar to a signed measure:
\[\chi(A \cup B) = \chi(A) + \chi(B) - \chi(A \cap B).\]
The field seeks to develop a theory of integration for constructible functions, similar to how regular calculus developed a theory of integration for measurable functions.
\begin{defn}
    Let $X$ be a definable function and $f: X \to \bb Z$ be a constructible function. The \textit{Euler integral} of $f$ is
    \[\int_X f(x) d\chi(x) \coloneqq \sum_{n = -\infty}^{\infty} n \chi(\{x \in X\ |\ f(x) = n\}).\]
    Note that this quantity is well-defined by the discussions in Definition~\ref{def: definable functions}(3). The \textit{Euler characteristic transform} of $f$ is defined as
    \[\ECT(f): \mathbb{S}^{n-1} \times \bb R \to \bb Z, (\nu, t) \mapsto \ECT(f)(\nu, t) = \int_{X} f(x) \mathbbm{1}_{X^v_t}(x) d\chi(x),\]
    where $X^v_t$ denotes the set $\{x \in X\ |\ v \cdot x \leq t\}$. For a definable subset $S \subseteq X$, We use $\ECT(S)$ to indicate the Euler characteristic transform of the indicator function on $S$.
\end{defn}

Here is an example of computation with the Euler characteristic transform.
\begin{exmp}
    Take $B_1(0) = \{ {x} \in \Rbb^n\ |\ |{x}| \leq 1\}$ to be the closed unit ball with respect to the $\ell_2$ norm. For any $\nu \in \bb S^{n-1}$, we have that $\ECT(B_1(0)(\nu, t) = 1$ if $-1 \leq t$ and $\ECT(B_1(0)(\nu, t) = 0$ if $t < -1$.
\end{exmp}

Euler calculus also enjoys its version of Fubini's Theorem.

\begin{thm}[Fubini's Theorem for Euler integrals]\label{thm::Fubini_Theorem}
    Let $f: X \to Y$ be a definable function between definable sets and $h: X \to \bb Z$ be a constructible function, then
    \[\int_X h(x) d\chi(x) = \int_Y \left( \int_{f^{-1}(y)} h(x) d\chi(x) \right) d\chi(y).\]
\end{thm}

\noindent A proof of Theorem~\ref{thm::Fubini_Theorem} may be found in Theorem 4.5 of \cite{curry2012euler}. Note that while the authors assumed $h$ to be compactly supported, the condition is not strictly required in the proof (see Page 5 of \cite{curry2012euler}). Theorem 1 of \cite{Gusein-Zade_2010} presents an explicit proof of Theorem~\ref{thm::Fubini_Theorem} for the case of semialgebraic sets without the assuming $h$ to be compactly supported, and the case for a general $o$-minimal structure follows similarly.

For convenience, we will also briefly explain what a Radon transform is and how it relates to the Euler characteristic transform.
\begin{defn}\label{def: radon transform}
    Let $(X, Y)$ be a pair of definable sets and $K \in \operatorname{CF}(X \times Y)$ (known as a \textit{kernel function}), then the \textit{Radon transform} is a function $R_K: \operatorname{CF}(X) \to \operatorname{CF}(Y)$ defined by
    \[(R_K h)(y) = \int_X h(x) K(x, y) d\chi(x), h(x) \in \operatorname{CF}(X), \text{ for all } y \in Y.\]
    In particular, when $Y = \bb S^{n-1} \times \bb R$ and $K$ is the indicator function on $\{(x, \nu, t) \in X \times Y \ |\ v \cdot x \leq t\}$, then $R_K$ is the ECT.
\end{defn}

\section{Euler Characteristic Transform with Hyperplanes}\label{sec:Hyperplanes}

In \cite{ghrist2018persistent}, the authors proved the following result on the injectivity of the Euler characteristic transform based on the Schapira inversion formula in \cite{schapira1995tomography}.

\begin{thm}[Theorem 1 of \cite{ghrist2018persistent}, Modified]\label{thm: hyperplane injective}
Let $X = \bb R^n$, $Y = \bb S^{n-1} \times \bb R$, $K \in \operatorname{CF}(X \times Y)$ be the indicator function on $\{(x, \nu, t) \in X \times Y \ |\ v \cdot x \leq t\}$ and $K' \in \operatorname{CF}(Y \times X)$ be the indicator function on $\{(\nu, t, x) \in Y \times X\ |\ v \cdot x \geq t\}$, then for any $h \in \operatorname{CF}(X)$, the following formula holds
\begin{align}\label{eq: inversion formula for ECT} 
(R_{K'} \circ R_{K}) h = (\mu - \lambda) h + \lambda (\int_X h d\chi) \mathbbm{1}_{X},
\end{align}
where $\mu = \chi(\bb S^{n-1})$ and $\lambda = 1$. Moreover, when restricted to the class of compactly supported functions on $X$, $R_K = \ECT$ is injective.
\end{thm}

In the original proof by the authors of \cite{ghrist2018persistent}, this formula is only stated in the case where $h$ is compactly supported. However, the formula still holds when $h$ is not compactly supported. Please see the Appendix (Section~\ref{sec: appendix}) for a proof of Equation~\ref{eq: inversion formula for ECT} without assuming that $h$ is compactly supported. While the $\ECT$ is injective on compactly supported constructible functions, it is not injective on $\operatorname{CF}(\mathbb{R}^n)$. We illustrate this with the following counter-example.
 
\begin{exmp}\label{ex: non-injective definable sets}
Let $X = \bb R^n$, $S_1 = \bb R^n$, and $S_2 = \emptyset$, then $\ECT(S_1)(\nu, t) = \ECT(S_2)(\nu, t)$ for all $(\nu, t) \in \bb S^n \times \bb R$. Indeed, the set $\{x \in \bb R^n\ |\ x \cdot \nu = t \}$ is definably homeomorphic to $\bb R^{n-1}$, and the set $\{x \in \bb R^n\ |\ x \cdot \nu > t\}$ is definably homeomorphic to $\bb R^n$. Hence, the additivity of Euler characteristic implies that
    \[\ECT(S_1)(\nu, t) = \chi(\{x \in \bb R^n\ |\ x \cdot \nu \leq t \}) = \chi(\bb R^{n-1}) + \chi(\bb R^n) = 0.\]
Hence $\ECT(S_1)$ is the zero function. On the other hand, the Euler characteristic of the empty set is always zero, so $\ECT(S_2)$ is also the zero function.
\end{exmp}


Fortunately, we can classify how non-injective is the $\ECT$ with the following theorem. From the theorem, we will also obtain a corollary that shows Example~\ref{ex: non-injective definable sets} is the only such counter-example for the case of definable sets.

\begin{thm}\label{thm:injectivity_hyperplane}
Let $f, g: \bb R^n \to \bb Z$ be constructible functions, then $\ECT(f) = \ECT(g)$ if and only if there exists some $c \in \bb Z$ such that
\[f(x) = \sum_{n = -\infty}^{+\infty} n \mathbbm{1}_{\{f^{-1}(n)\}}(x) \text{ and } g(x) = \sum_{n = -\infty}^{+\infty} (n + c) \mathbbm{1}_{\{f^{-1}(n)\}}(x).\]
In particular, suppose $f(x) = \mathbbm{1}_{S_1}(x)$ and $g(x) = \mathbbm{1}_{S_2}(x)$ for distinct definable sets $S_1, S_2 \subseteq \Rbb^n$ and $\ECT(f) = \ECT(g)$, then $S_1 = \Rbb^n$ and $S_2 = \emptyset$ up to renaming of variables.
\end{thm}

\begin{proof}
    Suppose $\ECT(f) = \ECT(g)$, then Equation~\ref{eq: inversion formula for ECT} implies that there exists integers $\mu \neq \lambda$ such that
    \[(\mu - \lambda) f(x) + \lambda \int_{\bb R^n} f(x) d\chi = (\mu - \lambda) g(x) + \lambda \int_{\bb R^n} g(x) d\chi,\]
    for all $x \in \bb R^n$. Hence, the difference $f(x) - g(x)$ is a constant integer, say $c$, and may be expressed as
    \begin{align}\label{eq: difference is constant}
        g(x) - f(x) = c \coloneqq \frac{\lambda}{\mu - \lambda} (\int_{\bb R^n} f(x) d\chi - \int_{\bb R^n} g(x) d\chi).
    \end{align}
    Since the images of constructible functions are finite, we can write $f(x) = \sum_{i = 1}^n a_i \mathbbm{1}_{A_i}(x)$ such that $A_i = f^{-1}(a_i)$ and $a_i$ ranges through the image of $f(x)$. Similarly, we can write $g(x) = \sum_{j = 1}^m b_j \mathbbm{1}_{B_j}(x)$ such that $B_j = g^{-1}(b_j)$ and $b_j$ ranges through the image of $g(x)$.

    Let $x \in A_i$, then $f(x) - g(x) = c$ by Equation~\ref{eq: difference is constant}. On the other hand $f(x) = a_i$, so $g(x) = b_j = a_i + c$ for some $b_j$ in the image of $g(x)$. Thus, the set function $\{a_1, ..., a_n\} \mapsto \{b_1, ..., b_m\}$ by $a_i \mapsto a_i + c$ is a well-defined injective set function. Similarly, the set function $\{b_1, ..., b_m\} \mapsto \{a_1, ..., a_n\}$ is also a well-defined inverse of the previous set function. Thus, we conclude that $n = m$ and $b_i = a_i + c$ for $i = 1, ..., n$ up to reordering.

     Now for any $x \in A_i$, $g(x) = c + f(x) = c + a_i = b_i$, so $x \in B_i$. Similarly for any $x \in B_i$, $f(x) = g(x) - c = a_i$, thus $x \in A_i$. Hence $A_i = B_i$. Thus, we conclude that
     \[f(x) = \sum_{i = 1}^n a_i \mathbbm{1}_{A_i}(x) \text{ and } g(x) = \sum_{i = 1}^n (a_i + c) \mathbbm{1}_{A_i}(x).\]
     This concludes the proof of the ``only if" direction.
    
    Conversely, suppose $f(x) = \sum_{n = -\infty}^{+\infty} n \mathbbm{1}_{\{f^{-1}(n)\}}(x)$ and $g(x) = \sum_{n = -\infty}^{+\infty} (n + c) \mathbbm{1}_{\{f^{-1}(n)\}}(x)$, then for any $(\nu, t) \in \bb S^{n-1} \times \bb R$, we will compute the difference of their respective Euler characteristic transforms.
    \begin{align*}
        \ECT(g)(\nu, t) - \ECT(f)(\nu, t) &= \sum_{n = -\infty}^{+\infty} \int_{\bb R^n} (n + c) \mathbbm{1}_{\{f^{-1}(n)\} \cap \{x \cdot \nu \leq t\}}(x) d\chi - \int_{\bb R^n} n \mathbbm{1}_{\{f^{-1}(n)\} \cap \{x \cdot \nu \leq t\}}(x) d\chi\\
        &= \sum_{n = -\infty}^{+\infty} \int_{\bb R^n} (n + c - n) \mathbbm{1}_{\{f^{-1}(n)\} \cap \{x \cdot \nu \leq t\}}(x) d\chi\\
        &= c \sum_{n = -\infty}^{+\infty} \int_{\bb R^n} \mathbbm{1}_{\{f^{-1}(n)\} \cap \{x \cdot \nu \leq t\}}(x) d\chi\\
        &= c \int_{\bb R^n} \mathbbm{1}_{\{x \cdot \nu \leq t\}}(x) d\chi\\
        &= c \ECT(\bb R^n)(\nu, t)\\
        &= 0,
    \end{align*}
    where the fourth line follows from the fact that the sets $\{f^{-1}(n)\}_{n \in \bb Z}$ form a finite (disregarding empty sets) partition of $\bb R^n$, and the sixth line follows from Example~\ref{ex: non-injective definable sets}.

    Finally, we will focus on the specific case that $f(x) = \mathbbm{1}_{S_1}(x)$ and $g(x) = \mathbbm{1}_{S_2}(x)$. Without loss of generality, we will assume that $S_1$ is non-empty. Since $\ECT(f) = \ECT(g)$, there exists some $c \in \bb Z$ such that $g(x) = 1 + c$ for all $x \in S_1$ and $g(x) = 0 + c$ for all $x \notin S_1$.

Since $S_1$ is not empty, then let $y$ be any point in $S_1$. $f(y) = 1$ implies that $g(y) = 1 + c$. If $g(y) = 1$, then $c = 0$ and $g(x)$ becomes the indicator function on $S_1$, which is a contradiction to the assumption that $f(x) \neq g(x)$. If $g(y) = 1 + c = 0$, then it follows that $c = -1$ and $g(x) = -1$ for all $x \notin S_1$. Since $g(x)$ takes values only between $0$ and $1$, this can occur only when $S_1 = \bb R^n$.

Thus, $S_1 = \bb R^n$ and $f(x)$ is the constant function with value $1$. Since $c = -1$, $g(x)$ is the constant function with value $0$, which implies that $S_2$ is the empty set.
\end{proof}





\section{Euler Characteristic Transform with Quadric Hypersurfaces}\label{sec:Quadric Hypersurfaces}

In this section, we will define the quadric Euler characteristic transform and prove several injectivity results on this transform. Before going into the $\QECT$ specifically, we will first discuss some results on Radon transforms in Section~\ref{subsec: kernel space} that will be useful in Section~\ref{subsec: QECT} (and the Appendix).

\subsection{Generalized Kernel Spaces}\label{subsec: kernel space}

Here is the general setup we will consider.
\begin{defn}\label{def: general kernel setup}
Let $X \subseteq \bb R^n$ be a definable set, $P \subseteq \bb R^k$ be a compact definable set (called the ``parameter space"), and $f: X \times P \to \bb R$ a definable function.  
\begin{enumerate}
    \item We define $K_f(x, (\xi, t)) \in \CF(X \times P \times \bb R)$ as the indicator function on $\{(x, (\xi, t)) \in X \times P \times \bb R\ |\ f(x, \xi) \leq t\}$ (the kernel function).
    \item We define $K'_f((\xi, t), x) \in \CF(P \times \bb R \times X)$ as the indicator function on $\{(\xi, t), x) \in P \times \bb R \times X\ |\ f(x, \xi) \geq t\}$ (the dual kernel function).
    \item We also define the fiber $K_{x, f} = \{(\xi, t) \in P \times \bb R\ |\ f(x, \xi) \leq t\}$ and the dual fiber $K'_{x', f} = \{(\xi, t) \in P \times \bb R\ |\ f(x', \xi) \geq t\}$.
\end{enumerate}
\end{defn}

Given a constructible function $h: X \to \bb Z$, we are interested in what the function $(R_{K'_f} \circ R_{K_f}) h$ is to be able to prove injectivity results similar to that of Theorem~\ref{thm: hyperplane injective}. We first prove a technical lemma.

\begin{lem}\label{lem::technical_euler_char}
    Let $x, x' \in X$.
    \begin{enumerate}
        \item $\chi(K_{x, f} \cap K_{x', f}') = \chi(\{\xi \in P\ |\ f(x', \xi) - f(x, \xi) \geq 0\})$.
        \item If $x = x'$, then $\chi(K_{x, f} \cap K_{x', f}') = \chi(P)$.
        \item As $(x, x')$ ranges through $X \times X$, the function $(x, x') \mapsto \chi(K_{x, f} \cap K_{x', f}')$ can only take on finitely many values $c_1, ..., c_n$. Furthermore, the preimage $S_i$ of each $c_i$ is a definable subset of $X \times X$.
    \end{enumerate}
\end{lem}

\begin{proof}
    For Lemma~\ref{lem::technical_euler_char}(1), we first rewrite the set $K_{x, f} \cap K_{x', f}'$ as follows.
    \begin{align*}
        K_{x, f} \cap K_{x', f}' &= \{(\xi, t) \in P \times \bb R\ |\ t \geq f(x, \xi) \text{ and } f(x', \xi) \geq t\}\\
        &= \{(\xi, t) \in P \times \bb R\ |\ f(x, \xi) \leq t \leq f(x', \xi)\}\\
        &= \{(\xi, t) \in P \times \bb R\ |\ t \in [f(x, \xi), f(x', \xi)] \}.
    \end{align*}
    By considering the definable homeomorphism $\phi: P \times \bb R \to P \times \bb R$ by $\phi(\xi, t) = (\xi, t - f(x, \xi))$, we have that
    \begin{align*}
        \chi(K_{x, f} \cap K_{x', f}') &= \chi(\phi(K_{x, f} \cap K_{x', f}))\\
        &= \chi(\{(\xi, t) \in P \times \bb R\ |\ t \in [0, f(x', \xi) - f(x, \xi)] \}).
    \end{align*}
    Since $P$ is compact and definable, the set $A \coloneqq \{(\xi, t) \in P \times \bb R\ |\ t \in [0, f(x', \xi) - f(x, \xi)] \}$ is compact and definable. Define the straight-line homotopy $H: A \times [0, 1] \to A$ as $H((\xi, t), s) = (\xi, (1 - s)t)$ for all $((\xi, t), s) \in A \times [0, 1]$, this produces a deformation retract of $A$ onto the set $\{\xi \in P\ |\ f(x', \xi) - f(x, \xi) \geq 0\}$, which preserves the Euler characteristic because both sets are compact and definable. The proof of Lemma~\ref{lem::technical_euler_char}(1) is thus completed.
    
    For Lemma~\ref{lem::technical_euler_char}(2), $x = x'$ implies that $f(x', p) - f(x, p) = 0$ for any $p \in P$. It then follows from Lemma~\ref{lem::technical_euler_char}(1) that $\chi(K_{x, f} \cap K_{x, f}') = \chi \{p \in P\ |\ 0 = 0\} = \chi(P)$.

    For Lemma~\ref{lem::technical_euler_char}(3), we implement Lemma~\ref{lemma:: finite euler char types} as follows. We define $S$ as the definable set
    \[S \coloneqq \{(x, x', \xi) \in X \times X \times P \subseteq \Rbb^{2n + k}\ |\ f(x', \xi) - f(x, \xi) \geq 0\}.\]
    In this case, $S_{(x, x')} = \{\xi \in \bb R^{k}\ |\ f(x', \xi) - f(x, \xi) \geq 0\}$. Thus, Lemma~\ref{lem::technical_euler_char}(3) follows directly from Lemma~\ref{lemma:: finite euler char types}.
\end{proof}

It is not generally true that for $x \neq x' \in P$, the value of $\chi(K_{x, f} \cap K'_{x', f})$ remains constant, as will be shown in the proof of Theorem~\ref{thm: quadric zero vector}. However, we can still compute the function $(R_{K'_f} \circ R_{K_f}) h$.

\begin{lem}\label{lem::formula_inversion_type}
    Following the context of Definition~\ref{def: general kernel setup} and Lemma~\ref{lem::technical_euler_char}(3) and fix $c_1 = \chi(P)$, then for any $x' \in X$,
    \[ (R_{K'_f} \circ R_{K_f}) h(x') = \chi(P) \int_{X} h(x) \mathbbm{1}_{S_1}(x, x') d\chi(x) + \sum_{i = 2}^n c_i \int_{X} h(x) \mathbbm{1}_{S_i}(x, x') d\chi(x).\]
    In particular, if $S_1 = \Delta$ is the diagonal of $X \times X$, then
    \[(R_{K'} \circ R_K) h(x') = \chi(P) h(x') + \sum_{i = 2}^n c_i \int_{X} h(x) \mathbbm{1}_{S_i}(x, x') d\chi(x).\]
\end{lem}

\begin{proof}
    By Lemma~\ref{lem::technical_euler_char}(3), we may write $\chi(K_{x, f} \cap K'_{x', f}) = \sum_{i = 1}^n c_i \mathbbm{1}_{S_i}(x, x')$ as a function of $x$ and $x'$. By Lemma~\ref{lem::technical_euler_char}, the diagonal $\Delta \subseteq X \times X$ is contained in exactly one of the $S_i$, say $S_1$. Then
    \begin{align*}
        (R_{K'_f} \circ R_{K_f}) h(x') &= \int_{P \times \bb R} K'(y, x') [\int_{X} h(x) K(x, y) d\chi(x)] d\chi(y) \\
        &= \int_{X} h(x) \left[\int_{P \times \bb R} K'_f(y, x') K_f(x, y) d\chi(y)\right] d\chi(x)\\
        &= \int_{X} h(x) \chi(K_{x, f} \cap K'_{x', f})\\
        &= \int_{X} h(x) \sum_{i = 1}^n c_i \mathbbm{1}_{S_i}(x, x') d\chi(x)\\
        &= c_1 \int_{X} h(x) \mathbbm{1}_{S_1}(x, x') d\chi(x) + \sum_{i = 2}^n c_i \int_{X} h(x) \mathbbm{1}_{S_i}(x, x') d\chi(x)\\
        &= \chi(P) \int_{X} h(x) \mathbbm{1}_{S_1}(x, x') d\chi(x) + \sum_{i = 2}^n c_i \int_{X} h(x) \mathbbm{1}_{S_i}(x, x') d\chi(x),
    \end{align*}
    where the second line follows from Theorem~\ref{thm::Fubini_Theorem}.
\end{proof}

\subsection{Quadric Euler Characteristic Transform}\label{subsec: QECT}
Let $V$ be the space of real $n \times n$ symmetric matrices equipped with a definable norm $||\bullet||_V$ (recall the notations in Definition~\ref{def: definable norm}). Recall that a general quadric surface is given by
\[x^T A x + v \cdot x = t\]
where $A \in V$, $v \in \Rbb^n$, and $t \in \Rbb$. Then it would seem that a natural definition of $\QECT$ on a constructible function $f: \bb R^n \to \Zbb$ would be
\[\QECT(f): S^V \times \bb S^{n-1} \times \Rbb \to \Zbb, \QECT(f)(A, v, t) = \int_{X} f(x) \mathbbm{1}_{X^{A, \nu}_t}(x) d\chi(x), \]
where $X^{A, \nu}_t$ denotes the set $\{x^T A x + v \cdot x \leq t\}$.

There is a question of whether our domain of choice is the best choice of domain. On one hand, $\QECT$ seems like a natural thematic generalization of the $\ECT$. However, there are no choices of $A \in S^V$ such that $||A||_V = 0$, so we cannot recover the $\ECT$ from this definition of the $\QECT$.

The domain of $\ECT(f)$ is $\bb S^{n-1} \times \Rbb$ avoids the degenerate case of the zero vector. Thus, what our more general $\QECT$ wants is to consider the case where $A$ and $v$ are both not identically zero. This suggests that we should think of the norm of $(A, v)$ as an element of $V \times \Rbb^n$, which we will refer to as the space $W$. We define the norm $||(A, v)||_W = ||A||_{V} + |v|$ for all $(A, v) \in W = V \times \bb R^n$.

Thus, we adjust our definition to the following.
\begin{defn}
Let $X \subseteq \bb R^n$ be a definable set and $f: X \to \bb Z$ be a constructible function, the quadric Euler characteristic transform of $f$ is the function $\QECT(f): S^W \times \bb R \to \bb Z$ defined by
    \[\QECT(f)(A, v, t) = \int_{X} \left( f(x) \mathbbm{1}_{X^{A, \nu}_t}(x) \right)  d\chi(x).\]
\end{defn}

\begin{exmp}
Let $S_1 = \bb R^n$, $S_2 = \emptyset$, and $I$ be the $n \times n$ identity matrix, then $\QECT(S_1)(I, 0, t) \neq \QECT(S_2)(I, 0, t)$. Hence, the $\QECT$ can tell the difference between $S_1$ and $S_2$ compared to Example~\ref{ex: non-injective definable sets}.
\end{exmp}

Now we will analyze a few properties of the $\QECT$. First of all, when $v = 0$ is fixed to be the zero vector, we note that the function $f \mapsto \{(A, t) \mapsto \QECT(f)(A, 0, t)\}$, which we will refer to as $\QECT(-, 0, -): \CF(\bb R^n) \to \CF(S^V \times \bb R)$, is not injective.

\begin{exmp}\label{ex: non-injective zero}
    Let $p \in \bb R^n$ be the vector whose components are all unity, $f(x) = \mathbbm{1}_{p}(x)$, and $g(x) = \mathbbm{1}_{-p}(x)$, then $\QECT(f)(A, 0, t) = \QECT(f)(A, 0, t)$ for all $(A, t) \in S^V \times \bb R$.
\end{exmp}

However, we can see that the failure to detect signs in Example~\ref{ex: non-injective zero} is the only such locus of non-injectivity with the following theorem.

\begin{thm}\label{thm: quadric zero vector}
Let $v = 0$ be fixed, the function $\QECT(-, 0, -): \operatorname{CF}(\bb R^n) \to \operatorname{CF}(S^V \times \bb R)$ is ``injective up to sign". More precisely, let $h: X \to \bb R^n$, we obtain an inversion formula reminiscent of the Schapira inversion formula,
    \[(R_{K'_f} \circ R_{K_f} h)(x') = (\mu - \lambda) \sum_{z \in \{+x', -x'\}} h(z) + \lambda (\int_X h d\chi) \mathbbm{1}_{X},\]
    where $\mu$ and $\lambda$ are distinct integers.
\end{thm}

\begin{proof}
Since $v = 0$, $||A||_V = 1 - |v| = 1$. Thus, the function $(A, t) \mapsto \QECT(f)(A, 0, t)$ has domain $S^V \times \bb R$. Following the setup of Definition~\ref{def: general kernel setup}, we choose $X = \bb R^n$, $P = S^V$, and $f: X \times P \to \bb R$ to be the function $(x, A) \mapsto x^T A x$. By Lemma~\ref{lem::technical_euler_char} and the property that $A$ is symmetric, $\chi(K_{x, f} \cap K'_{x', f}) = \chi(\{A \in P\ |\ (x')^T A (x') - (x)^T A (x) \geq 0\}) = \chi(\{A \in P\ |\ (x' + x)^T A (x' - x) \geq 0\})$.

If $x = x'$ or $x = -x'$, then $(x' + x)^T A (x' - x) = 0$ so $\chi(K_{x, f} \cap K'_{x', f}) = \chi(S^V)$.

Otherwise, suppose $x \notin \{\pm x'\}$, then consider the function $\phi: V \to \bb R$ given by $\phi(A) = (x' + x)^T A (x' - x)$. Since $(x + x')$ and $(x' - x)$ are both non-zero vectors, $\phi$ is a surjective linear transformation, and hence $\ker(\phi)$ is a codimension $1$ linear subspace of $V$. $\ker(\phi)$ inherits a natural norm from $V$ and hence $S^V \cap \ker(\phi) = S^{\ker(\phi)}$.

Since all norms on a finite-dimensional real vector space are equivalent, the map $\psi: S^V \to \bb S^{\dim V - 1}$ by $\psi(x) = \frac{x}{|x|}$ is a definable homeomorphism whose restriction to $S^{\ker(\phi)}$ gives a homeomorphism between $S^{\ker(\phi)}$ and the $\ell_2$ unit sphere of $\ker(\phi)$.

By Alexander duality, $\Tilde{H}^0(S^V \setminus S^{\ker(\phi)}) \cong \Tilde{H}^{\dim V - 1 - 1}(S^{\ker(\phi)}) \cong \bb Z$, so $S^V \setminus S^{\ker(\phi)}$ has two connected components. Since $\phi$ is an odd function, we will denote the two connected components as $S^V_+$ and $S^V_-$ corresponding to the locus where $\phi$ is positive and negative respectively. Hence, the map $\psi: S^V \to \bb S^{\dim V - 1}$ brings $\ker(f) \cup S^V_+$ homeomorphically to a closed hemisphere of $\bb S^{\dim V - 1}$. Thus, the set $\ker(f) \cup S^V_+$ is compact and contractible, so we conclude that $\chi(K_{x, f} \cap K'_{x', f}) = 1$.

Let $\mu = \chi(P)$ and $\lambda = 1$, by Lemma~\ref{lem::formula_inversion_type}, we can write
\[(R_{K'_f} \circ R_{K_f}) h(x') = \mu \int_{\bb R^n} h(x) \mathbbm{1}_{\pm \Delta}(x, x') d\chi(x) + \lambda \int_{\bb R^n} h(x) \mathbbm{1}_{\{X \times X - \pm \Delta\}}(x, x') d\chi(x),\]
where $\mathbbm{1}_{\pm \Delta}(x, x') = 1$ if $x \in \{\pm x'\}$ and is $0$ otherwise. Finally, we can furthermore simplify the expression as
\begin{align*}
    (R_{K'_f} \circ R_{K_f}) h(x') &=  \mu \int_{\{\pm x'\}} h(x) d\chi(x) + \lambda \int_{\bb R^n - \{\pm x'\}} h(x) d\chi(x)\\
    &= (\mu - \lambda) \int_{\{\pm x'\}} h(x) d\chi(x) + \lambda \int_{\bb R^n} h(x) d\chi(x)\\
    &= (\mu - \lambda) \sum_{z \in \{\pm x'\}} h(z) +  \lambda \int_{\bb R^n} h(x) d\chi(x).
\end{align*}
This concludes the proof of Theorem~\ref{thm: quadric zero vector}.
\end{proof}

Theorem~\ref{thm: quadric zero vector} examined what happens to the $\QECT$ when its vector component is fixed. We are also interested in what happens to the $\QECT$ when its matrix component is fixed.

\begin{thm}\label{thm: quadric interpolation}
Let $B_R(0)$ denote the closed ball of radius $R$ for some $R \geq 0$. Fix the norm on $V$ to be the operator norm $||\bullet||_{op}$, and let $A \in V$ be fixed such that $||A||_{op} < \frac{1}{1 + 2R^2}$. Define the function $\QECT(A, -, -): \operatorname{CF}(B_R(0)) \to \operatorname{CF}(P \times \bb R)$ by $f \mapsto \{(v, t) \mapsto \QECT(f)(A, v, t)\}$, where $P$ is the sphere of radius $1 - ||A||_{op}$ in $\bb R^n$. Then, $\QECT(A, -, -)$ is injective.
\end{thm}

\begin{rem}\label{rem: interpolation}
Theorem~\ref{thm: quadric interpolation} interpolates between Theorem~\ref{thm: quadric zero vector} and Theorem~\ref{thm: hyperplane injective} in the following sense.
\begin{enumerate}
    \item When $||A||_{op} = 1$, the inequality $1 < \frac{1}{1 + 2R^2}$ does not hold for any value of $R$. This is reflective of the fact that $(A, t) \mapsto \QECT(A, 0, t)$ is only injective up to signs in Theorem~\ref{thm: quadric zero vector} and Example~\ref{ex: non-injective zero}.
    \item When $||A||_{op} = 0$, $A$ is the zero matrix and the $\QECT$ becomes the usual $\ECT$, which is injective by Theorem~\ref{thm: hyperplane injective} no matter what $R$ is. This reflects the fact that the inequality $0 < \frac{1}{1 + 2 R^2}$ is satisfied for any value of $R$.
\end{enumerate}
The requirement for the norm on $V$ to be the operator norm is not strictly necessary. The same statement holds for definable norm $||\bullet||_V$ on $V$ that satisfies the property $||x^T A x||_V \leq ||A||_V |x|^2$. Common examples include the Frobenius norm and the nuclear norm. Furthermore, by adjusting the constant $\frac{1}{1 + 2R^2}$ appropriately, similar statements for any definable norm on $V$ will hold.
\end{rem}

Now we will prove Theorem~\ref{thm: quadric interpolation}.

\begin{proof}[Proof of Theorem~\ref{thm: quadric interpolation}]
Following the setup of Definition~\ref{def: general kernel setup}, we choose $X = B_R(0)$, $P = \{v \in \bb R^n\ |\ |v| = 1 - ||A||_{op}\}$, and $f: X \times P \to \bb R$ to be the function $(x, v) \mapsto x^T A x + v \cdot x$. By Lemma~\ref{lem::technical_euler_char} and the property that $A$ is symmetric, $\chi(K_{x, f} \cap K'_{x', f}) = \chi(\{v \in P\ |\ (x')^T A (x') + v \cdot (x') - (x)^T A (x) - v \cdot x \geq 0\}) = \chi(\{v \in P\ |\ (x' + x)^T A (x' - x) + v \cdot (x' - x) \geq 0\})$.

If $x = x'$, then Lemma~\ref{lem::technical_euler_char} tells us that $\chi(K_{x, f} \cap K'_{x', f}) = \chi(P)$. Otherwise, if $x \neq x'$, we can consider the function $\phi: \bb R^n \to \bb R$ by $\phi(v) = (x' + x)^T A (x' - x) + v \cdot (x' - x)$. Since $x' - x$ is not the zero vector, $\phi$ is a surjective affine map and $\ker(\phi)$ is a hyperplane in $\bb R^n$. Furthermore, the vector $\frac{(x' - x)}{|x - x'|}$ is a unit normal vector to $\ker(\phi)$, and $\ker(\phi)$ may be written as the sum
\[\ker\{v \mapsto v \cdot (x' - x)\}  - \frac{(x' + x)^T A (x' - x)}{|x - x'|} (x' - x).\]
This is because for any $v - \frac{(x' + x)^T A (x' - x)}{|x - x'|} (x - x')$ in the sum above,
\begin{align*}
    \phi(v - \frac{(x' + x)^T A (x' - x)}{|x - x'|} (x' - x)) &= (x' + x)^T A (x' - x)+ [v - \frac{(x' + x)^T A (x' - x)}{|x - x'|} (x - x')] \cdot (x' - x)\\
    &= (x' + x)^T A (x' - x) + 0 - \frac{(x' + x)^T A (x' - x)}{|x - x'|} |x' - x|\\
    &= (x' + x)^T A (x' - x) -(x' + x)^T A (x' - x)\\
    &= 0.
\end{align*}
Now we observe that
\begin{align*}
    |((x' + x)^T A (x' - x)| &= |(x')^T A (x') - (x)^T A (x)|\\
    &\leq |x^T A x| + |(x')^T A (x')|\\
    &\leq ||A||_{op} |x|^2 + ||A||_{op} |x'|^2\\
    &\leq ||A||_{op} (2 R^2)\\
    &< \frac{2R^2}{1 + 2R^2}\\
    &= 1 - \frac{1}{1 + 2R^2}\\
    &< 1 - ||A||_{op},
\end{align*}
where the fifth and last line both follow from the assumption $||A||_{op} < \frac{1}{1 + 2R^2}$. The radius of $P$ is $1 - ||A||_{op}$, so $|(x + x')^T A (x - x')| < 1 - ||A||_{op}$ implies that $\{v \in P\ |\ \varphi(v) \geq 0\}$ is a compact and contractible subset of $P$. Thus, $\chi(K_{x, f} \cap K'_{x', f}) = \chi(\{v \in P\ |\ \varphi(v) \geq 0\}) = 1$.

Thus, by Lemma~\ref{lem::formula_inversion_type}, for all $x' \in \bb R^n$,
\begin{align*}
         R_{K'_f} \circ R_{K_f} h &= [1 + (-1)^{n-1}] h(x') + \int_{\bb R^n} h(x) \mathbbm{1}_{\{(x, x') \notin \Delta \}} d\chi(x)\\
         &= [1 + (-1)^{n-1}] h(x') + \int_{\bb R^n - \{x'\}} h(x) d\chi(x)\\
         &= [1 + (-1)^{n-1}] h(x') + \int_{\bb R^n} h(x) d\chi(x) - \int_{\{x'\}} h(x) d\chi(x)\\
         &= [(-1)^{n-1}] h(x') + \int_{\bb R^n} h(x) d\chi(x).
\end{align*}
Since $h(x) \in \operatorname{CF}(B_R(0))$ has compact support, $\QECT$ would determine the value of $\int_{\bb R^n} h(x) d\chi(x)$. Thus, the function $\QECT(A, -, -): \operatorname{CF}(B_R(0)) \to \operatorname{CF}(\bb S^{n-1} \times \bb R)$ is injective.
\end{proof}

Finally, we will also discuss some auxiliary properties of the $\QECT$.
\begin{prop}\label{prop:: auxiliary QECT}
    Let $f: \bb R^n \to \Zbb$ be a constructible function, then
    \begin{enumerate}
        \item $\QECT(f)$ takes only finitely many values as $(A, \nu, t)$ runs through $S^W \times \Rbb$. 
        \item For a fixed $(A, \nu) \in S^W$, the function $t \mapsto \QECT(f)(A, \nu, t)$ is right continuous.
    \end{enumerate}
\end{prop}

\begin{proof}Proposition~\ref{prop:: auxiliary QECT}(1) is a direct application of Lemma~\ref{lemma:: finite euler char types} whose proof of similar to that of Lemma~\ref{lem::technical_euler_char}(3). Thus, the proof is omitted here. Proposition~\ref{prop:: auxiliary QECT}(2) is a direct application of Theorem 3.1 of \cite{ji2023euler}.
\end{proof}

\section{Appendix}\label{sec: appendix}

\subsection{Inversion Formula of $\ECT$ without Compact Support}

Here we reprove the inversion formula of $\ECT$ in \cite{ghrist2018persistent} without the assumption that $h: \bb R^n \to \bb Z$ is compactly supported. The proof will be similar to that of \cite{ghrist2018persistent}.

\begin{proof}[Proof of Equation~\ref{eq: inversion formula for ECT}]
    Following the setup of Definition~\ref{def: general kernel setup}, we choose $X = \bb R^n$, $P = \bb S^{n-1}$, and $f: X \times P \to \bb R$ to be the function $(x, \nu) \mapsto x \cdot \nu$. By Lemma~\ref{lem::technical_euler_char}, we can compute that
    \begin{align*}
        \chi(K_{x, f} \cap K_{x'}', f) &= \chi(\{\nu \in \bb S^{n-1}\ |\ f(x', \nu) - f(x, \nu) \geq 0\})\\
        &=  \chi(\{\nu \in \bb S^{n-1}\ |\ (x' - x)\cdot \nu \geq 0\})\\
        &= \begin{cases}
            \chi(\bb S^{n-1}), x = x'\\
            \chi(\bb S^{n-1}_+), x \neq x'
        \end{cases}\\
        &= \begin{cases}
            1 + (-1)^{n-1}, x = x'\\
            1, x \neq x'
        \end{cases},
    \end{align*}
     where $\bb S^{n-1}_+$ denotes the closed upper hemisphere of $\bb S^{n-1}_+$. By Lemma~\ref{lem::formula_inversion_type}, for all $x' \in \bb R^n$,
     \begin{align*}
         (R_{K'} \circ R_{K}) h &= [1 + (-1)^{n-1}] h(x') + \int_{\bb R^n} h(x) \mathbbm{1}_{\{(x, x') \notin \Delta \}} d\chi(x)\\
         &= [1 + (-1)^{n-1}] h(x') + \int_{\bb R^n} h(x) d\chi(x) - \int_{\{x'\}} h(x) d\chi(x)\\
         &= [(-1)^{n-1}] h(x') + \int_{\bb R^n} h(x) d\chi(x).
     \end{align*}
     The proof of Equation~\ref{eq: inversion formula for ECT} is completed.
\end{proof}

\bibliographystyle{abbrvnat}
\bibliography{references}

\end{document}